\documentclass[journal,comsoc]{IEEEtran}
\usepackage[T1]{fontenc}
\usepackage{amssymb}
\usepackage{amsmath}
\usepackage{amsthm}
\usepackage{amsfonts}
\usepackage{amsthm}
\usepackage{amsmath}
\usepackage{amsfonts}
\usepackage{amssymb}
\usepackage{bbm}
\usepackage{graphicx}
\usepackage{hyperref}
\usepackage{algorithm,algpseudocode}
\usepackage{mathtools}
\usepackage{dsfont}
\usepackage{verbatim}
\usepackage{subcaption}
\usepackage{float}
\usepackage{eucal}
\usepackage{balance}
\usepackage{bm}
\usepackage{amssymb}
\usepackage{amsmath}
\usepackage{amsthm}
\usepackage{amsfonts}
\usepackage{hyperref}
\usepackage{bbold}
\usepackage{mathtools}
\usepackage{dsfont}
\usepackage{verbatim}
\usepackage{subcaption}
\usepackage{lipsum}
\usepackage{cuted}
\usepackage{xfrac}
\newtheorem{theorem}{Theorem}

\newtheorem{lemma}[theorem]{Lemma}

\title{Capacity and Performance Analysis of RIS-Assisted Communication Over Rician Fading Channels}

\author{
\IEEEauthorblockN{Chandradeep~Singh, Chia-Hsiang~Lin, and Kamal~Singh}
\vspace*{-0.8cm}
}

\begin{document}
\maketitle
\begin{abstract}
This paper investigates two performance metrics, namely ergodic capacity and symbol error rate, of mmWave communication system assisted by a reconfigurable intelligent surface (RIS). We assume independent and identically distributed (i.i.d.) Rician fadings between user-RIS-Access Point (AP), with RIS surface consisting of passive reflecting elements. First, we derive a new unified closed-form formula for the average symbol error probability of generalised $M$-QAM/$M$-PSK signalling over this mmWave link. We then obtain new closed-form expressions for the ergodic capacity with and without channel state information (CSI) at the AP.
\end{abstract}
\vspace{-1.1em}
\section{Introduction}
\IEEEPARstart{T}{he} traditional cellular networks, operating at sub 1-GHz and sub 6-GHz, are severely spectrum-limited to meet the high-bandwidth requirements of future wireless data services. For the recently launched 5G and proposed future 6G systems, wireless communication technologies are under development for the Millimeter Wave (mmWave) spectrum (30 GHz-300 GHz) due to large bandwidth availability at these frequencies. To overcome the fundamental challenge of very-high attenuation at mmWave frequencies, the communication is assisted by employing highly directional antennas. Unfortunately, these directional links can be easily blocked by obstacles leading to outage. One feasible solution is to focus the transmitted beam on a reconfigurable intelligent surface (RIS), which then redirects it ‘smartly’ to the intended user~\cite{Wu2020}. The RIS is a promising new technology that can assist communication in challenging 5G/6G wireless scenarios. The key feature of RIS technology is that it consists of nearly passive low cost reflecting elements with reconfigurable parameters to provide virtual line of sight (LoS) path between two communicating nodes to bypass the obstacle~\cite{Basar2019}. We remark that active elements of metasurfaces can also be achieved in recent nanotechnology literature~\cite{Lin2021}. A sufficiently large RIS 
with $N$ reflecting elements can provide a beamforming gain of the order of $\mathcal{O}(N^2)$~\cite{Wu2019}. Most notably, the RIS aided mmWave systems perform better with small number of antennas than existing mmWave wireless technologies such as massive multiple-input and multiple-output systems, multi-antenna amplify-and-forward relay systems etc.  while significantly reducing the implementation costs~\cite{Wu2019,Boulogeorgos2020}.

In this work, we will analyze the performance of an RIS-assisted mmWave communications in terms of the well known performance metrics namely symbol error probability and ergodic capacity~\cite{Boulogeorgos2020},~\cite{Salhab2021}. Ergodic capacity is an important information-theoretic performance metric which provides the highest rate of communication that can be supported by a communication channel with arbitrarily small probability of error. On the other hand, for existing contemporary communication schemes in use, symbol error probability is a more relevant performance measure that indicates the reliability of the communication link. Various channel models have been proposed that encompasses the statistical properties of a typical RIS-assisted mmWave communication link. 
Specifically, cascaded Rician channel models are found to be more suitable for RIS-assisted mmWave communication because both of the AP-RIS and RIS-user mmWave links have a strong/dominant LoS component in addition to non-LoS scattering components~\cite{Gao2021}-\cite{Wu_2020}. The probability density function (PDF) of the cascaded Rician channel is given by an infinite series representation with terms involving modified Bessel functions of the second kind~\cite{Donoughue2012}, which enforces the need for a simple asymptotic approximation for analytical tractability. In~\cite{Yildirim2021}, the RIS-assisted mmWave channel's distribution is approximated by a non-central chi-square distribution to derive unified symbol error probability bounds for the mmWave indoor and outdoor communication scenarios. In~\cite{Han2019}, a simple adaptive discrete phase shift design of the RIS elements is proposed that achieves spectral efficiency close to ergodic capacity in a RIS-assisted mmWave communication link. In another related work, maximal ratio transmission (MRT) over RIS-assisted mmWave communication link is analyzed and the phase shifts are optimized to minimize the outage probability~\cite{Guo2020}. In~\cite{Tao2020}, asymptotic expressions for outage probability and capacity bounds are derived for the RIS-assisted communication link. In~\cite{Jia2020}, authors optimize the phase shifts of the RIS elements to maximize ergodic rate with instantaneous and statistical CSI in RIS-assisted mmWave communication channel. Outage probability and ergodic rate expressions have been derived for single-antenna RIS-assisted communication systems in independent Rayleigh fadings in~\cite{Kudathanthirige2020}. For the less-relevant RIS-assisted mmWave communication assuming independent and identically distributed (i.i.d.) Rayleigh fadings, improved approximation of the overall channel distribution is obtained in~\cite{Yang2020}. Most recent performance results on the RIS-assisted communications subjected to i.i.d. Rician fadings are derived in terms of outage probability, symbol error probability and ergodic capacity \emph{without} CSI in~\cite{Salhab2021}.  

While the aforementioned works study the capacity \emph{without} CSI, no explicit characterization of the ergodic capacity \emph{with} CSI at the transmitter side in RIS-aided communication in independent Rician fadings is available in the current literature. Our main contributions in this paper are two-fold:
\begin{itemize}
\item First, we derive a novel exact closed-form expression of the RIS-assisted mmWave channel's capacity subjected to independent Rician fadings and with full CSI at the transmitter. Additionally, for the \emph{without} CSI case, we obtain a new closed-form capacity expression in terms of Meijer's $G$ function which is computationally convenient compared to the derived capacity result in~\cite{Salhab2021}. We further characterize the asymptotic behaviour of the capacity without CSI in the low-SNR regime which is missing in~\cite{Salhab2021}.
\item We then propose a new unified symbol error probability (SEP) formula for the generalized $M$-QAM/$M$-PSK signalling over the RIS-assisted mmWave communication link subjected to i.i.d. Rician fadings. The proposed 
formula is simple and very compact compared to the SEP expression derived in~\cite{Salhab2021}.
\end{itemize}

The organization of the paper is as follows: RIS-assisted mmWave system model is described
in Section~\ref{sec:sys_model}. The derivations of the capacity and SEP performance measures are 
detailed in Section~\ref{Perfom_analys}. Numerical results are presented and further discussed 
in Section~\ref{sec:ne} to demonstrate the effectiveness of the RIS-assisted communications. Finally, 
conclusions are drawn in Section~\ref{sec:con}.
\section{System Model}\label{sec:sys_model}
We consider a single-antenna based mmWave communication link assisted by an RIS as shown in Fig.~\ref{fig:sys_model}. The RIS consists of a total of $N$ reflecting elements indexed by  $l$. We denote 
by $h_l$ the complex-valued channel gain from AP to the reflecting element $l$ at RIS, and by $g_l$ the complex valued channel gain from the $l$-th RIS element to the user. Phase shift induced by $l$-th RIS element is denoted by $\phi_l$. The channel vector for the AP-RIS link is given by $\bm{h}=\left[h_1,\ldots,h_N\right]^T$. Similarly, the channel vector for the RIS-user link is given by $\bm{g}=\left[g_1,\ldots,g_N\right]^T$. The phase shift matrix of the RIS is given as $\bm{\phi} \triangleq \mbox{diag}\left(\phi_1,\ldots,\phi_N\right)\in \mathbb{C}^{N{\times}N}$. $P_t(\cdot)$ is power scheme employed at the transmitter. Denoting the transmitted symbol by $x$, then the received signal at user is given by
\begin{equation}
y =\sqrt{P_t}\bm{g}^T\bm{\phi}\bm{h}x+n, \label{eq:channelmodel}
\end{equation}
where $n$ is the additive white Gaussian noise (AWGN) $n \sim \mathcal{CN}(0,N_0)$. Denoting the complex channel coefficient $h_l$ as $h_l := \alpha_le^{\theta_l}$ where $\alpha_l$ and $\theta_l$ are the amplitude and phase shift respectively. Similarly, we express the channel coefficient $g_l$ as $g_l := \beta_le^{\theta'_l}$, where $\beta_l$ is the amplitude and $\theta'_l$ is the phase shift. Further, we assume that the envelopes of first hop channels ${\alpha_1,\ldots,\alpha_N}$ are i.i.d. Rician random variable with shape parameter $K_1$ and scale parameter $\Omega_1$. Likewise, envelopes of the second hop channels ${\beta_1,\ldots,\beta_N}$ are assumed to be i.i.d. Rician random variables with shape parameter $K_2$ and scale parameter $\Omega_2$. We consider that the channel vectors $\bm{h}$ and $\bm{g}$ are statistically independent. Let $v_i^2$ denote the power in the LoS component and $2\sigma_i^2$ denote the power in the non-line-of-sight (NLoS) component. The shape and scale parameter are denoted by $K_i=\frac{v_i^2}{2\sigma_i^2}$ and $\Omega_i=v_i^2+2\sigma_i^2$ respectively.
\begin{figure}[h]
\begin{centering}
\includegraphics[scale=0.4]{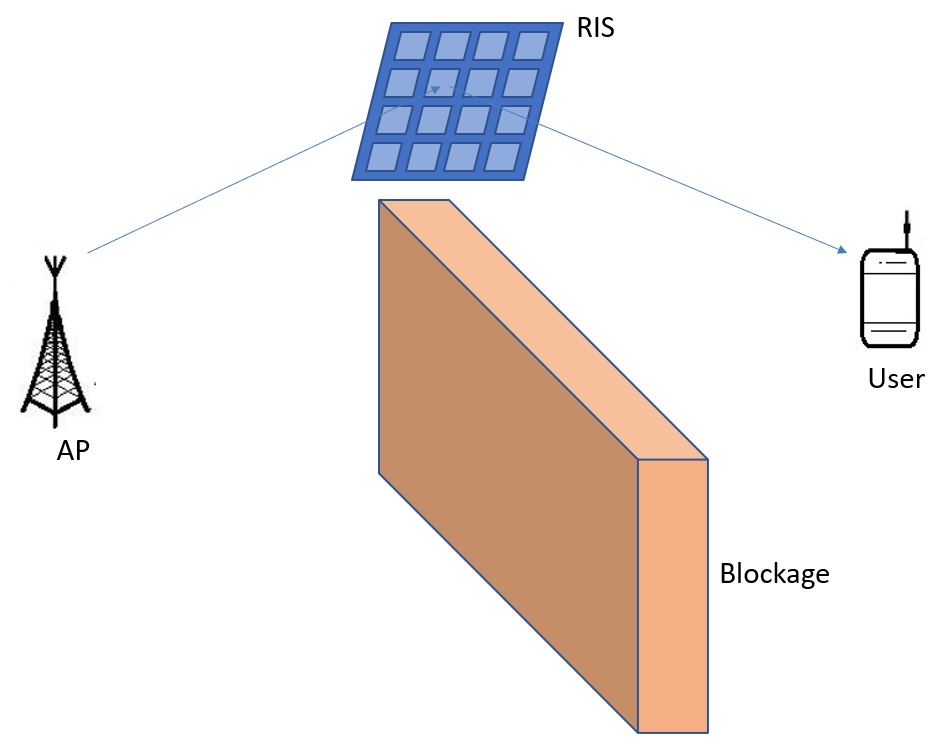}
\caption{Illustration of RIS-Assisted Communications.}
\label{fig:sys_model}
\end{centering}
\vspace{-0.65em}
\end{figure}

When transmitter expends power at a constant level $P_{\mathrm{avg}}$, according to~\cite{Basar2019}, the maximum end-to-end SNR is given by 
\begin{equation}\label{eq:snr}
\gamma =\overline{\gamma}\left(\sum_{l=1}^{N}\alpha_l\beta_l\right)^2\\ 
\phantom{x} \mbox{where }\phantom{x}\overline{\gamma}=\frac{P_{\mathrm{avg}}}{N_0}, \nonumber
\end{equation}
with the implicit assumption of channel information available at the RIS to provide optimal beamforming gain by optimizing its phase shift matrix $\bm{\phi}$. Let $\xi=\sum_{l=1}^{N}\xi_l$ where $\xi_l=\alpha_l\beta_l$. The PDF of $\xi_l$ can be obtained as~\cite[Eq.~(24)]{Donoughue2012}
\begin{align}\label{eq:pdf123}
&f_{\xi_l}\left(y\right) = \nonumber\\
&\sum_{j=1}^{\infty}\sum_{i=1}^{\infty}\frac{y^{i+j+1}\left(4K_2^iK_1^j(\Omega_1\Omega_2)^{\frac{(i+j+2)}{2}}\right)K_{j-i}(2y\sqrt{\Omega_1\Omega_2})}{(i!)^2(j!)^2\exp(K_1K_2)},
\end{align}
where $K_{\nu}(\cdot)$ is the modified Bessel function of the second kind and order $\nu$. The PDF in~\eqref{eq:pdf123} is an infinite series that has to be approximated for simulations and analysis purpose. For practical applications,  we use the first term of the Laguerre expansion to approximate the PDF of the sum of $N$ i.i.d. non-negative random variables \cite[Sec. 2.2.2]{Primak2004}. The mean and variance of the cascaded Rician random variable $\xi_l$ are computed as
\vspace{-0.5em}
\begin{flalign}\label{eq:mean} 
&\phantom{xxxxxxxxxxxxxxx}\mathbbm{E}[\xi_l] = \mathbbm{E}[\alpha_l]\mathbbm{E}[\beta_l],&\\[0.35em]
&\text{where}\phantom{xxxxxxx}\mathbbm{E}[\alpha_l] = \frac{1}{2}\sqrt{\frac{\pi \,  \Omega_1}{K_1+1}}L_{1/2}(-K_1),&\nonumber\\
&\text{and}\phantom{xxxxxxxxx}\mathbbm{E}[\beta_l] = \frac{1}{2}\sqrt{\frac{\pi \, \Omega_2}{K_2+1}}L_{1/2}(-K_2).&\nonumber
\end{flalign}
Here, $L_{1/2}(\cdot)$ denotes the Laguerre polynomial, i.e.,  $L_{1/2}(x)=e^{x/2}\left[(1-x)I_0\left(\frac{-x}{2}\right)-xI_1\left(\frac{-x}{2}\right)\right]$,
where $I_{\nu}(\cdot)$ is the modified Bessel function of the first kind and order $\nu$. Further simplification of~\eqref{eq:mean} gives
\begin{align}
\mathbbm{E}[\xi_l] &= \frac{\pi e^{-\frac{(K_1+K_2)}{2}}}{4}\sqrt{\frac{\Omega_1\Omega_2}{(K_1+1)(K_2+1)}} \nonumber \\
&\phantom{x}\,\times \left[(K_1+1)I_0\left(\frac{K_1}{2}\right)+K_1I_1\left(\frac{K_1}{2}\right)\right] \nonumber \\
&\phantom{x}\,\times \left[(K_2+1)I_0\left(\frac{K_2}{2}\right)+K_2I_1\left(\frac{K_2}{2}\right)\right].
\end{align}
Notice that $\mathbbm{E}[ \xi_l^2 ] =\mathbbm{E}[\alpha_l^2]\mathbbm{E}[\beta_l^2]=\Omega_1\Omega_2$. We now compute the variance as follows:
\begin{align} \label{eq:var}
\mathrm{Var}\left(\xi_l\right) &=\mathbbm{E}[\xi_l^2]-\mathbbm{E}^2[\xi_l]=\Omega_1\Omega_2-\mathbbm{E}^2[\xi_l].
\end{align}

\begin{lemma}
The PDF and cumulative density function (CDF) of the end-to-end SNR are respectively given by 
\begin{equation}\label{eq:pdf_snr}
f_{\gamma}\left(\gamma\right) \simeq \frac{\gamma^{\frac{a-1}{2}}\exp\left(-\frac{\sqrt{\gamma}}{b\sqrt{\overline{\gamma}}}\right)}{2b^{a+1}\Gamma\left(a+1\right)\overline{\gamma}^{\frac{a+1}{2}}}, \mbox{ and }
\end{equation}
\begin{equation}\label{eq:pdf_snr}
F_{\gamma}\left(\gamma\right) \simeq \frac{\gamma\left(a+1,\frac{\sqrt{\gamma}}{b\sqrt{\overline{\gamma}}}\right)}{\Gamma\left(a+1\right)},
\end{equation}
where $\gamma(\cdot,\cdot)$ is the lower incomplete Gamma function given in \cite[Eq.~8.350.1]{Gradshteyn2000}, $\Gamma(\cdot)$ is the Gamma function, 
$a=\frac{\mathbbm{E}^2[\xi_l]}{\mathrm{Var}(\xi)}-1$ and $b=\frac{\mathrm{Var}(\xi)}{\mathbbm{E}[\xi_l]}$.
\end{lemma}
\begin{proof}
As $\xi$ is the sum of i.i.d non-negative random variables, therefore, by using only the first term of Laguerre expansion given in~\cite[Sec. 2.2.2]{Primak2004}, we can approximate the PDF of $\xi$ as 
\begin{equation}\label{eq:pdf_si}
f_{\xi}\left(y\right) \simeq \frac{y^a}{b^{a+1}\Gamma\left(a+1\right)}\exp\left(-\frac{y}{b}\right).
\end{equation}
$\mathbbm{E}[\xi] =N \cdot \mathbbm{E}[\xi_l]$ and $\mathrm{Var}(\xi)=N \cdot \mathrm{Var}(\xi_l)$. The results then simply follows from the relation $\gamma=\overline{\gamma}\, \xi^2$.
\end{proof}
Note that for the distribution $f_{\xi}\left(y\right)$,~\eqref{eq:pdf_si} is significantly more convenient rather than utilizing the infinite series expansion given in~\eqref{eq:pdf123}. In the next section, we will analyze the performance of the RIS-assisted mmWave communication utilizing the approximate distribution derived in~\eqref{eq:pdf_si}.

\section{Performance Analysis: Symbol Error Probability \& Ergodic Capacity}\label{Perfom_analys}
\subsection{Average Symbol Error Probability}
For generalized $M$-QAM/$M$-PSK signalling schemes, the unified average symbol error probability (ASEP) expression is of the form $ASEP = \mathbbm{E} [p \, Q(\sqrt{ 2 q \gamma})] $ where $Q(\cdot)$ is the Gaussian Q-function, and $p$ and $q$ are modulation-specific constants; for more details, we refer the interested reader to~\cite{Alouini1999}. This expectation can be written using the CDF as follows~\cite{McKay2007}:
\begin{equation}\label{eq:err_prob}
ASEP=\frac{p\sqrt{q}}{2\sqrt{\pi}}\int_{0}^{\infty}\frac{\exp(-q\gamma)}{\sqrt{\gamma}}F_{\gamma}\left(\gamma\right)d\gamma.
\end{equation}

\begin{theorem}
For the RIS aided communication, the unified ASEP formula is given by
\begin{equation}\label{eq:err_probclos}
ASEP=\frac{2^{(a-1)}p}{\pi\Gamma(a+1)} G_{3,4}^{2,3} \left(\frac{1}{4q\overline{\gamma}b^2}\bigg|_{\frac{a+1}{2},\frac{a+2}{2},0,\frac{1}{2}}^{\frac{1}{2},\frac{1}{2},1}\right),
\end{equation}
where $G_{p,q}^{m,n}(\cdot)$ is Meijer's G function as given in \cite[Sec. 9.3]{Gradshteyn2000}. 
\end{theorem}
\begin{proof}
Eq. \eqref{eq:err_prob} can be written as 
\begin{equation}\label{eq:err_prob1}
ASEP=\frac{p\sqrt{q}}{2\sqrt{\pi}\Gamma\left(a+1\right)}\int_{0}^{\infty}\frac{e^{-q\gamma}}{\sqrt{\gamma}}\gamma\left(a+1,\frac{\sqrt{\gamma}}{b\sqrt{\overline{\gamma}}}\right)d\gamma.
\end{equation}
As given in \cite[Sec. 8.4.16]{Prudnikov1998}, the lower incomplete Gamma function can be expressed as $\gamma(v,x)=G_{1,2}^{1,1}\left( x\big|_{v,0}^{1}\right)$. Therefore,
\begin{equation}\label{eq:err_prob2}
ASEP=\frac{p\sqrt{q}}{2\sqrt{\pi}\Gamma\left(a+1\right)}\int_{0}^{\infty}\frac{e^{-q\gamma}}{\sqrt{\gamma}}G_{1,2}^{1,1}\left( \frac{\sqrt{\gamma}}{b\sqrt{\overline{\gamma}}}\bigg|_{a+1,0}^{1}\right)d\gamma.
\end{equation}
Using identity given in~\cite[Sec. 2.24.3.1]{Prudnikov1998} as \eqref{eq:identity}  at top of next page,
\begin{figure*}[t!]
\normalsize
\noindent\rule{18.2cm}{0.2pt}
\begin{align}\label{eq:identity}
\int_{0}^{\infty}x^{\alpha-1} e^{-\sigma x} &G_{p,q}^{m,n}\left( \omega x^{l/k}\bigg|_{b_1,\ldots,b_m,\ldots,b_q}^{a_1,\ldots,a_n,\ldots,a_p}\right)dx
\nonumber \\
&=\frac{k^{\mu}l^{\alpha-1}\sigma^{-\alpha}}{(2\pi)^{\frac{l-1}{2}+c^*(k-1)}} 
G_{kp+l,kq}^{km,kn+l}\left( \frac{\omega^k l^l}{\sigma^l k^{k(q-p)}}\bigg|_{\frac{b_1}{k},\ldots, \frac{k+b_1-1}{k},\ldots,\frac{b_m}{k},\ldots, \frac{k+b_m-1}{k},\ldots,\frac{b_q}{k},\ldots, \frac{k+b_q-1}{k}}^{\frac{1-\alpha}{l},\ldots,\frac{l-\alpha}{l},\frac{a_1}{k},\ldots, \frac{k+a_1-1}{k},\ldots,\frac{a_n}{k},\ldots, \frac{k+a_n-1}{k},\ldots,\frac{a_p}{k},\ldots, \frac{k+a_p-1}{k}}\right),
\end{align}
\mbox{where} $\phantom{x}\mu \,\,=\,\,\sum_{j=1}^q  \,b_j-\sum_{j=1}^p \,a_j+\frac{p-q}{2}+1$ $\phantom{x}$    \mbox{and} $\phantom{x}$ $c^* \,\,= \,\,m+n-\frac{p+q}{2}$.\hfill
\vspace{-0.25em}
\noindent\rule{18.2cm}{0.2pt}
\end{figure*}
we get the desired result. Therefore, we have completed the proof.
\end{proof}
Note that the derived expression for ASEP is compact  and computationally very convenient compared to the ASEP expression derived in~\cite{Salhab2021} as the Meijer's $G$ function is available in typical numerical computing environments (e.g., MATLAB, Mathematica).
\subsection{Ergodic Capacity}
\textbf{Without CSI:} When the Access Point (AP) transmits at fixed power level, the channel's ergodic capacity is computed as follows~\cite[Eq.~(8)]{Goldsmith1997}:
\begin{equation}\label{eq:erg_cap}
 \widehat{C} =\frac{1}{\ln2}\int_{0}^{\infty}\ln(1+\gamma)f_{\gamma}\left(\gamma\right)d\gamma.
\end{equation}
\begin{theorem}\label{thm:no_csi}
For the RIS aided communication channel, the ergodic capacity without CSI is given by
\begin{equation}\label{eq:nocsi_cap}
 \widehat{C} =\frac{2^{a}}{\ln2\sqrt{\pi}\Gamma(a+1)} G_{4,2}^{1,4} \left(4\overline{\gamma}b^2\bigg|_{1,0}^{\frac{-a}{2},\frac{-a+1}{2},1,1}\right).
\end{equation} 
\end{theorem}
\begin{proof}
As given in \cite[Sec. 8.4.6]{Prudnikov1998}, $\ln(1+x)$ can be expressed as $\ln(1+x)=G_{2,2}^{1,2}\left( x\big|_{1,0}^{1,1}\right)$. Therefore, Eq.~\eqref{eq:erg_cap} can be written as
\begin{eqnarray} 
\begin{split}
 \widehat{C}  &=\frac{1}{2\ln2(b^{a+1})\Gamma\left(a+1\right)\overline{\gamma}^{\frac{a+1}{2}}} \times\\
&\phantom{xxxxxxxxxx}\int_{0}^{\infty}\gamma^{\frac{a-1}{2}}\exp\left(-\frac{\sqrt{\gamma}}{b\sqrt{\overline{\gamma}}}\right)G_{2,2}^{1,2}\left( \gamma\big|_{1,0}^{1,1}\right)d\gamma.  \nonumber
\end{split}
\end{eqnarray}
Changing the variable of integration as $\gamma=t^2$, we get
\begin{eqnarray} 
\begin{split}
 \widehat{C}  &=\frac{1}{\ln2(b^{a+1})\Gamma\left(a+1\right)\overline{\gamma}^{\frac{a+1}{2}}} \\
& \int_{0}^{\infty}t^a\exp\left(-\frac{t}{b\sqrt{\overline{\gamma}}}\right)G_{2,2}^{1,2}\left( t^2\big|_{1,0}^{1,1}\right)dt.  \nonumber
\end{split}
\end{eqnarray}

\noindent Using the identity given in Eq.~\eqref{eq:identity}, we get the desired capacity expression.
\end{proof}

We remark that the low-SNR regime can be very relevant in mmWave communications as SNR per degree of freedom can become low due to very large bandwidth available at mmWave frequencies~\cite{Verdu2004}. We can also analyze the asymptotic capacity behaviour in the low-SNR regime, i.e., as $\overline{\gamma}\to 0$. As $\overline{\gamma}\to 0$, the series expansion of the Meijer's $G$ function in~\eqref{eq:nocsi_cap} is dominated by the first term as follows~\cite{Gradshteyn2000}:
\begin{equation}
G_{4,2}^{1,4} \left(4\overline{\gamma}b^2\bigg|_{1,0}^{\frac{-a}{2},\frac{-a+1}{2},1,1}\right) \approx 4b^2 \Gamma\left(\frac{a+3}{2}\right)\Gamma\left(\frac{a+4}{2}\right)\overline{\gamma}.
\end{equation}
Hence, the asymptotic low-SNR capacity expression without CSI is given by
\begin{equation}\label{eq:nocsi_asym}
 \widehat{C}  \approx \frac{2^{a+2}b^2\Gamma(\frac{a+3}{2})\Gamma(\frac{a+4}{2})\overline{\gamma}}{\ln2\sqrt{\pi}\Gamma(a+1)}.
\end{equation}
It shows that in the low-SNR regime, the ergodic capacity without CSI varies linearly with the average SNR $\overline{\gamma}$. 

\textbf{With CSI:} The optimal power strategy with instantaneous CSI, say $P_t(\cdot)$, is the well-known water-filling scheme given as $\frac{P_t(\gamma)}{P_{\mathrm{avg}}} = \frac{1}{\gamma_{0}}-\frac{1}{\gamma}$ for $\gamma>\gamma_{0}$ and zero otherwise~\cite{Goldsmith1997}. The cutoff $\gamma_0$ is determined from the average power constraint $\int_{\gamma_0}^{\infty}\left(\frac{1}{\gamma_{0}}-\frac{1}{\gamma}\right)f_{\gamma}\left(\gamma\right)d\gamma=1$.
\begin{theorem}\label{thm:csi}
For the RIS aided communication channel, the ergodic capacity with instantaneous CSI is given by
\begin{equation}\label{eq:csi_cap}
 C =\frac{2^{a}}{\ln2\sqrt{\pi}\Gamma(a+1)} G_{4,2}^{0,4} \left(\frac{4b^2}{\gamma_0}\bigg|_{0,0}^{\frac{-a}{2},\frac{-a+1}{2},1,1}\right).
\end{equation} 
\end{theorem}
\begin{proof}
Ergodic capacity with water-filling power scheme is given by~\cite[Eq.~(7)]{Goldsmith1997}
\begin{align*}
C &= \frac{1}{\ln2} \int_{\gamma_0}^{\infty}\ln\left( \gamma / \gamma_0 \right)f_{\gamma}\left(\gamma\right)d\gamma \\
  &= \frac{\gamma_0}{\ln2}   \int_{0}^{\infty} \ln (x)H(x-1) f_{\gamma} \left(\gamma_0 x \right) dx.
\end{align*}
where the last equality follows due to the $\gamma / \gamma_0 = x$ variable substitution and noting that $H(x-1)$ is the unit-step function. Following \cite[Sec. 8.4.6]{Prudnikov1998}, $\ln(x)H(x-1)$ can be expressed as $\ln(x)H(x-1)=G_{2,2}^{0,2}\left( x\big|_{0,0}^{1,1}\right)$. Now, following similar steps as used in Theorem~\ref{thm:no_csi}, we get the desired capacity result.
\end{proof}

We have not considered asymptotic capacity analysis with CSI due to space limitations. To the extent of our knowledge, the derivations~\eqref{eq:err_probclos},~\eqref{eq:nocsi_cap},~\eqref{eq:nocsi_asym} and~\eqref{eq:csi_cap} are new and compact. 
\section{Numerical Results \& Discussion}
\label{sec:ne}
In this section, we analyze the numerical results to illustrate the capacity and error performance of RIS aided mmWave communications in i.i.d. Rician fadings. Fig.~\ref{fig:asep} illustrates the effect of increasing the Rician factor on the BPSK modulation's $(p=q=1)$ error probability performance. Expectedly, as the Rician factor improves, the diversity gain increases as indicated by the increasing slope (magnitude only) of the ASEP curves.
\begin{figure}[H]
\begin{centering}
\includegraphics[scale=0.75]{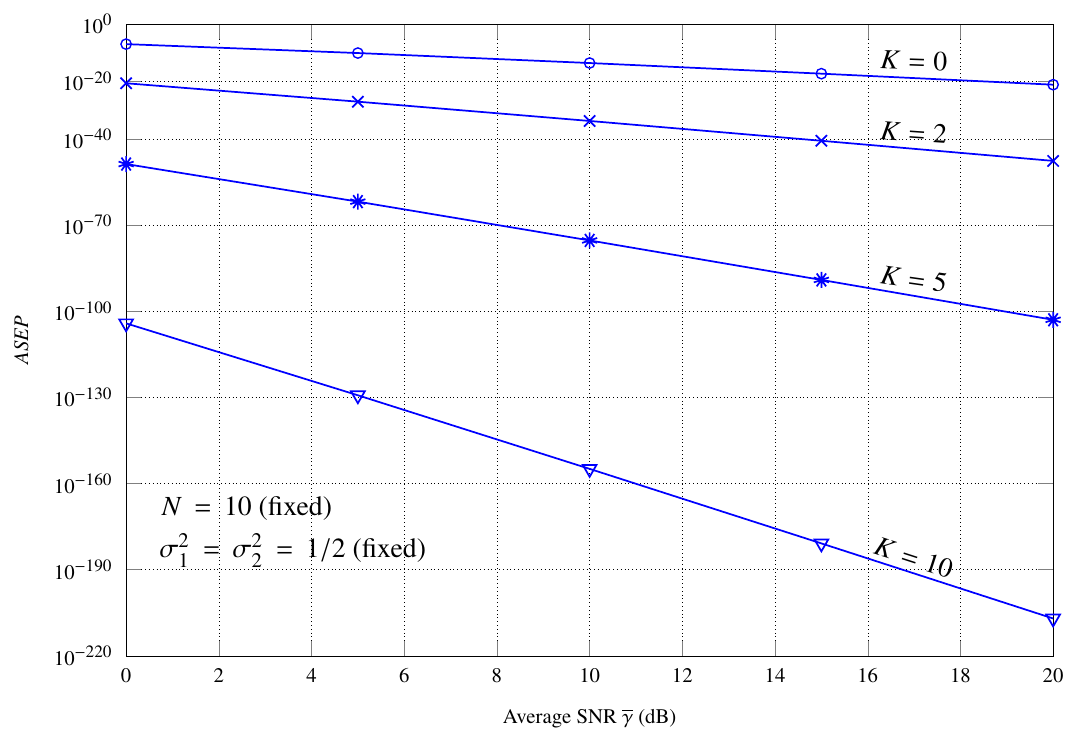}
\caption{ASEP vs. SNR for different values of $K_1=K_2=K$.}
\label{fig:asep}
\end{centering}
\end{figure}
\begin{figure}[H]
\begin{centering}
\includegraphics[scale=0.75]{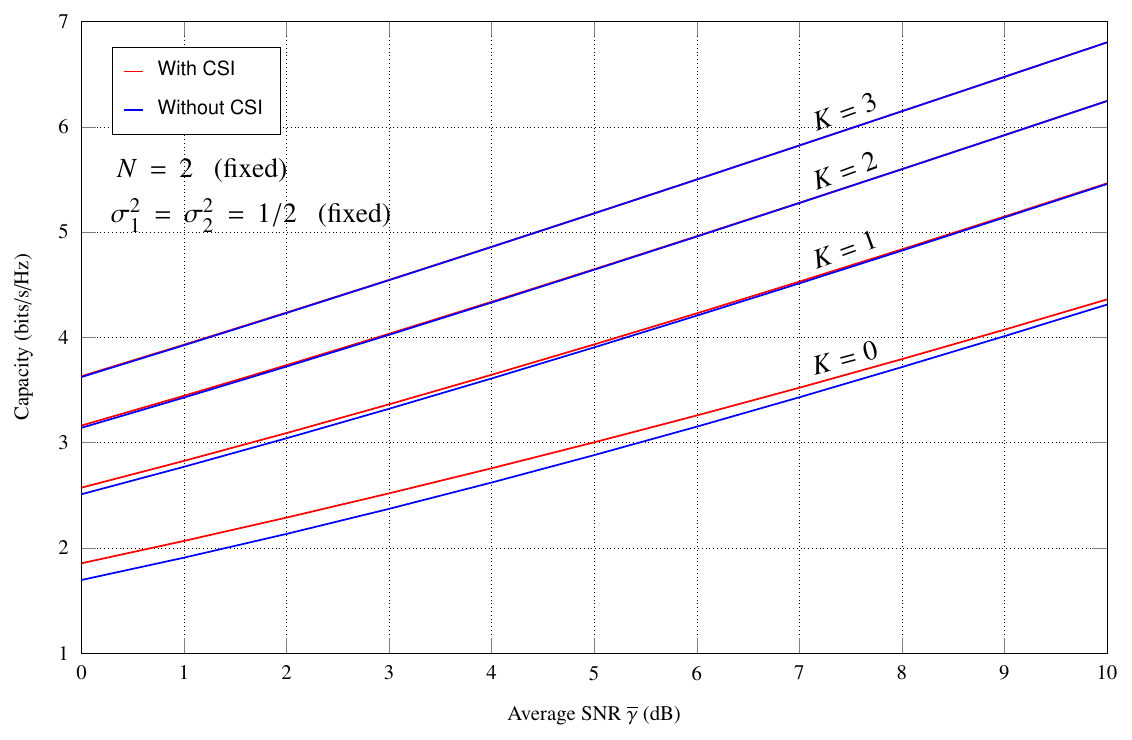}
\caption{Capacity vs. SNR for different values of $K_1=K_2=K$.}
\label{fig:element_2}
\end{centering}
\end{figure}
\begin{figure}[H]
\begin{centering}
\includegraphics[scale=0.75]{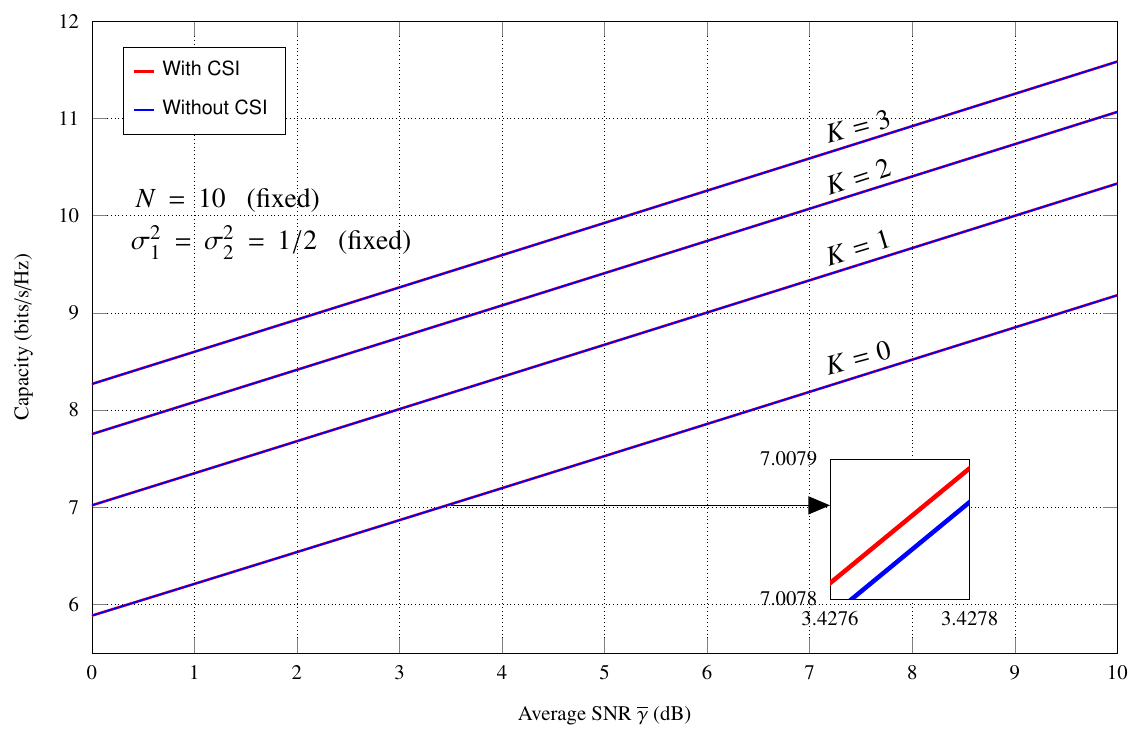}
\caption{Capacity vs. SNR for different values of $K_1=K_2=K$.}
\label{fig:element_10}
\end{centering}
\end{figure}
For the numerical results in Figs.~\ref{fig:element_2} and~\ref{fig:element_10}, we consider $\sigma_1^2=\sigma_2^2=\sfrac{1}{2}$, $K_1=K_2=K$, and $N=2$ and $N=10$ respectively. Numerical results for the RIS-assisted communication with small number of elements in Fig.~\ref{fig:element_2} show that the CSI based water filling power  allocation improves the capacity in the low-SNR regime. Fig.~\ref{fig:element_2} also shows that for low-Rician factor $K$ values, power control compensates for the large scattering in NLoS directions and improves the throughputs. However, for the higher values of $K$, transmission without power control performs as good as water-filling power allocation. The effect of increasing the reflecting elements $N$ is to compensate the scattering losses due to low Rician $K$-factor by providing higher beamforming gains as illustrated in Fig.~\ref{fig:element_10}. Further, Fig.~\ref{fig:element_10} also demonstrates that for large RIS-assisted single-antenna mmWave communication system, transmission without power control is nearly optimal and thus, CSI to the transmitter side is not necessary. We emphasize here that the number of reflecting elements on the RIS are as per the system requirements. If the system requirements are not very stringent then it is sufficient to employ low number of reflecting elements as suggested in~\cite{Salhab2021}. 

In Fig.~\ref{fig:asympt}, we illustrate the accuracy of the asymptotic capacity expression without CSI in the low-SNR regime. We find that the asymptotic capacity starts approaching the exact capacity as at reasonably low-SNRs; e.g., for the chosen settings, the convergence begins around -30 dB of average SNR. Notice that the capacity gap due to CSI and without CSI increases as $\overline{\gamma} \to 0$. Fig.~\ref{fig:asympt} also demonstrates the important role of CSI based transmit-power adaptation in the low SNR regime to enhance the achievable throughputs.
\begin{figure}[H]
\begin{centering}
\includegraphics[scale=0.78]{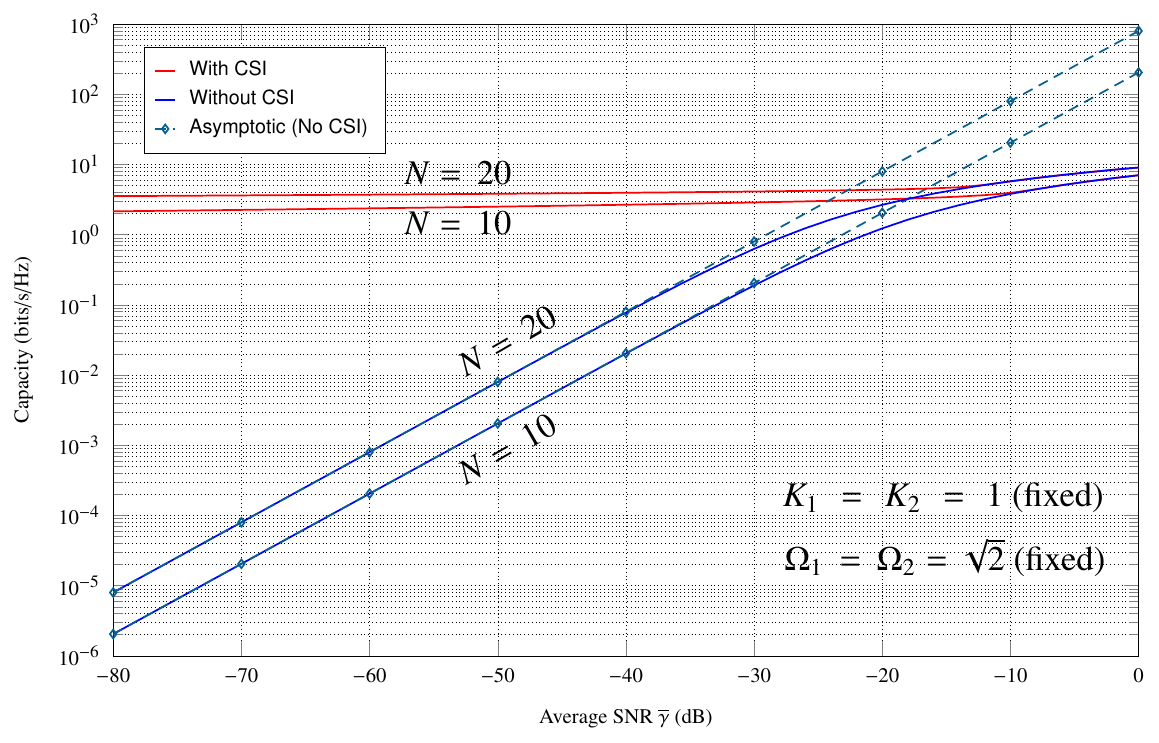}
\caption{Capacity vs. SNR in the low-SNR regime.}
\label{fig:asympt}
\end{centering}
\end{figure}
\section{Conclusion}
\label{sec:con}
In this work, we have derived a novel closed-form capacity formula for the RIS-assisted mmWave communication channel in i.i.d. Rician fadings with full CSI at the transmitter. Then, a new closed-form expression was derived for the capacity without CSI, which was further characterized in the low-SNR regime. Next, we present a new compact ASEP formula in closed-form for the generalized $M$-QAM/$M$-PSK signalling over the RIS-assisted communication. Numerical results showed that CSI based power adaptation compensates the high scattering losses in an RIS-assisted communications. Finally, we noted that for a mmWave link assisted by a large RIS, the transmission without power control is nearly optimal for a wide range of Rician fading conditions.
\section*{Acknowledgements}
\label{sec:ack}
This work was supported partly by the ``NCKU 90 and Beyond'' initiative at National Cheng Kung University (NCKU), Taiwan, partly by the Young Scholar Fellowship Program (Einstein Program) of Ministry of Science and Technology (MOST) in Taiwan under Grant MOST 110-2636-E-006-026, and partly by the Higher Education Sprout Project of the Ministry of Education (MOE) to the Headquarters of University Advancement at NCKU.

\balance

\end{document}